\documentclass[12pt,a4paper]{article}
\usepackage{amsfonts}
\usepackage{amsmath}
\usepackage{epsfig,amssymb,amsmath,latexsym}
\usepackage[all]{xy}

\def\qed{\unskip\nobreak\hfill$\Box$\par\addvspace{\medskipamount}}

\def\b{{\beta}}

\def\l{{\lambda}}

\def\s{{\sigma}}

\def\p{{\rho}} 
\def\td{\mathrm{d}}

\newcommand{\be}{\begin{equation}}
\newcommand{\ee}{\end{equation}}
\newcommand{\bea}{\begin{eqnarray}}
\newcommand{\eea}{\end{eqnarray}}
\newcommand{\beas}{\begin{eqnarray*}}
\newcommand{\eeas}{\end{eqnarray*}}

\newtheorem{theorem}{Theorem}[section]
\newtheorem{definition}[theorem]{Definition}
\newtheorem{proposition}[theorem]{Proposition}
\newtheorem{corollary}[theorem]{Corollary}
\newtheorem{lemma}[theorem]{Lemma}
\newtheorem{remark}[theorem]{Remark}
\newtheorem{example}[theorem]{Example}
\newtheorem{examples}[theorem]{Examples}
\newtheorem{foo}[theorem]{Remarks}
\newenvironment{Example}{\begin{example}\rm}{\end{example}}

\newenvironment{Remark}{\begin{remark}\rm}{\end{remark}}

\newenvironment{proof}{\addvspace{\medskipamount}\par\noindent{\it Proof}.}
{\unskip\nobreak\hfill$\Box$\par\addvspace{\medskipamount}}

\newcommand{\IR}{\int_{\mathbb{R}^k\setminus\{0\}}}

\newcommand{\E}[1]{{\mathbb{E}}\left[#1\right]}

\DeclareMathOperator{\essinf}{ess\,inf}
\DeclareMathOperator{\argmin}{arg\,min}

\def\F{\mathcal{F}}
\def\R{\mathbb{R}}

\begin{document}

\title{Representation Results for Law Invariant Recursive Dynamic Deviation Measures and Risk Sharing}

\author{Mitja Stadje
	\footnote{{Faculty of Mathematics and Economics,
	Ulm University}, mitja.stadje@uni-ulm.de
\newline
\newline
{\em (2010) AMS Classification.} 60H30, 91B30, 91B06, 91B70, 46N10.
}}

\maketitle
\begin{abstract}
In this paper we analyze a dynamic recursive extension of the (static) notion of a deviation measure and its properties. We study distribution invariant deviation measures and show that the only dynamic deviation measure which is law invariant and recursive is the variance. We also solve the problem of optimal risk-sharing generalizing classical risk-sharing results for variance through a dynamic inf-convolution problem involving a transformation of the original dynamic deviation measures.


\end{abstract}

\section{Introduction}

The traditional way of thinking about risk, playing a crucial role in most 
fields involved with probabilities, is to measure risk as the deviation of the 
random outcomes from the longtime average, i.e., to measure risk for instance 
as the variance or the standard deviation involved.
This is in particular the case for portfolio choice theory where almost the 
complete standard finance literature simply describes portfolio selection as 
the choice between return (mean) and risk (variance). For stock prices in a 
continuous-time setting risk is also often identified 
with volatility, i.e., as the local standard deviation on an incremental time 
unit. \\

However, variance penalizes positive deviations from the mean in the same way as negative deviations, which in many contexts is not suitable. Furthermore, computing the variance or the standard deviation is mainly justified by its nice analytical, computational and statistical properties but is a ad-hoc procedure and it is not clear if not better methods could be used. To overcome these shortfalls, Rockafellar {\emph et al.}~(2002) developed a general axiomatic framework for static deviation measures; see also Rockafellar {\emph et al.}~(2006a,2006b,2006c,2007,2008), M\"arket and Schultz~(2005), Sarykalin {\emph et al.}~(2008), Grechuk {\emph et al.}~(2009), Righi and Ceretta~(2016) or Righi~(2017).

This work was inspired by the axiomatic construction of coherent and convex 
risk measures given in Artzner {\emph et al.}~(1999,2000), F\"ollmer and Schied~(2002) 
and Frittelli and Rosazza Gianin~(2002). (Coherent or convex risk measures 
describe minimal capital reserves a financial institution should hold in order 
to be `safe'.) As Artzner {\emph et al.}~(2000) gave an axiomatic 
characterization of 
capital reserves these works give an axiomatic framework for deviation 
measures. \\

This theory of generalized deviation measures can be extended to a dynamic setting using the conditional variance formula, see Pistorius and Stadje~(2017), in the same spirit as convex risk measures have been extended to a dynamic setting using the tower property. For the latter, see for instance Barrieu and El Karoui~(2004,2005,2009), Artzner {\emph et al.}~(2004), Riedel~(2004), Rosazza Gianin~(2006),  Cheridito {\emph et al.}~(2006), Delbaen~(2006), Kl\"oppel and Schweizer~(2007), Jiang~(2008), Delbaen {\emph et al.}~(2008), Bion Nadal~(2009), Pelsser and Stadje~(2014) and Elliott et al. (2015). \\

In the first part of the paper we study distribution invariant deviation measures. For distribution invariant convex risk measures, Kupper and Schachermayer~(2009) showed, building on results of Gerber~(1974), that the so called entropic risk measure is the only convex risk measure satisfying the tower property, see also Goovaerts and Vylder~(1979) and Kaluszka and Krzeszowiec~(2013). The entropic risk measure arises as the negative certainty equivalent of a decision maker with an exponential utility function, see for instance F{\"o}llmer and  Schied~(2002,2004). Although there is an abundance of \emph{static} distribution invariant deviation measures we show that the only dynamic deviation measure which is law invariant and recursive is the variance. Interestingly, it is known also for other contexts that there is a close relationship between the variance and the entropic risk measure (or equivalently the use of an exponential utility function). For instance it is well known in the economics literature that the mean-variance principle can be seen as a second order Taylor approximation to the entropic risk measure. Furthermore, both induce preferences  which are invariant  under shifts of wealth and lead to the same optimal portfolios under normality assumptions, see for example 
Cochrane~(2009) for an overview. Moreover, it is shown for instance in Pelsser and Stadje~(2014) that in a Brownian filtration applying mean-variance  recursively  over an infinitesimal small time interval is equivalent to applying the entropic risk measure recursively  over an infinitesimal small time interval. This paper adds to these results showing that the entropic risk measure and the variance are the only distribution invariant risk measures, which naturally extend to continuous time under dynamic consistency conditions. \\

Subsequently, we then, after discussing some examples, analyze risk-sharing of payoffs between two agents for distribution invariant and non-distribution invariant dynamic deviation measures, generalizing classical risk-sharing results for variance. For coherent, convex and more general risk measures, static and dynamic risk-sharing was studied in Jouini et al. (2008), Acciaio (2007), Barrieu and El Karoui~(2004,2005,2009), Filipovic and Svindland (2008), Carlier et al. (2012), Heath and Ku (2004),  Tsanakas (2009), Dana and Le Van (2010), Mastrogiacomo and Rosazza Gianin (2015), Weber (2017) and Embrechts, et al. (2018). These works go back at least to the seminal work by Borch (1962).  
We solve the risk sharing problem through a dynamic inf-convolution problem involving a transformation of the original driver functions of the dynamic deviation measures involved. \\

This paper is structured as follows. Section \ref{Section Setting} introduces the setting and the basic concepts and definitions. Section \ref{Section Distribution Invariant Deviation Measures} analyzes distribution invariant dynamic deviation measures, while section \ref{Section Examples and Applications} solves the optimal risk-sharing problem.

\section{Setting} \label{Section Setting}

Formally, we consider from now on a filtered, completed, right-continuous 
probability space
$(\Omega,\F,(\F_{t})_{t\in[0,T]},P)$.
Throughout the text, equalities and inequalities between random variables are
meant to hold $P$-almost surely (a.s.); two random variables are identified if
they are equal $P$-a.s. For $t\in[0,T]$, we define $L^{2}(\F_t)$ as the space 
of $\F_t$-measurable random  variables $X$ such that $\E{X^2}< \infty.$ 
\\


\noindent{\bf (Conditional) deviation measures.} Dynamic deviation measures are given in terms of 
conditional deviation measures, which are in turn a conditional version of the notion of (static) deviation measure as in Rockafellar {\em et al.}~(2006a) that we describe next. On a filtered probability space $(\Omega,\F,(\F_{t})_{t\in [0,T]},\mathbb P)$, where $T>0$ denotes the horizon, consider the (risky) positions described by elements in $L^{2}(\F_t)$, $t\in[0,T]$, the space of $\F_t$-measurable random variables $X$ such that $\E{|X|^2}< \infty$). By $L^{2}_{+}(\F_t)$, $L^{\infty}(\F_t)$ and $L^{\infty}_+(\F_t)$ are denoted the subsets of non-negative, bounded and non-negative bounded elements in $L^{2}(\F_t)$.
\begin{definition}
For any given $t\in [0,T]$, $D_t:L^2(\F_T)\to L^2_+(\F_t)$ is called 
an {\em $\mathcal F_t$-conditional  generalized deviation measure}
if it is {\em normalized} ($D_t(0) = 0$)
and the following
properties are satisfied:
\begin{itemize}
\item[{\rm(D1)}]
{\em Translation Invariance}: $D_t(X + m) = D_t(X)$ for any  $m\in
L^\infty(\F_t)$;
\item[{\rm(D2)}] {\em Positivity}: $D_t(X) \geq 0$ for any $X\in L^2(\mathcal F_T),$ and $D_t(X) = 0$ if and only if $X$ is $\F_t$-measurable;
\item[{\rm(D3')}] {\em Subadditivity}: $D_t(X + Y) \leq D_t(X) + D_t(Y)$ for any $X, Y \in L^2(\mathcal F_T)$;
\item[{\rm(D4')}] {\em Positive Homogeneity}: $D_t(\lambda X) = \lambda D_t(X)$ for any $X \in L^2(\F_T)$ and $\lambda \in L^\infty_+(\F_t)$.
\end{itemize}
\end{definition} 
If $\mathcal F_0$ is trivial, $D_0$ is a deviation measure in the sense of Definition~1 in 
Rockafellar {\em et al.}~(2006a). The value $D_t(X)=0$, we recall, corresponds to the riskless state of no uncertainty, and axiom (D1) and can be interpreted as the requirement that adding to a position $X$ a constant (possibly interpreted as cash)  should change the overall deviation. Furthermore, it follows similarly as in Rockafellar {\em et al.} (2006a) 
that, if $D$ satisfies (D3')--(D4'), 
(D1) holds if and only if $D_t(m)=0$ for any $m\in L^2(\F_t)$.  In other words, constants do not carry any risk.
Moreover, it is well known that if (D4') holds, (D3') is equivalent to
\begin{itemize}
\item[(D3)] Conditional Convexity: 
For any $X,Y\in L^2(\F_T)$ and 
any $\lambda\in L^\infty(\F_t)$ that satisfies $0\leq\lambda\leq 1$ 
$$D_{t}(\lambda X+(1-\lambda)Y)\leq \lambda D_{t}(X)+(1-\lambda)D_{t}(Y). $$
\end{itemize}

 
\begin{definition}
	For any given $t\in [0,T]$, $D_t:L^2(\F_T)\to L^2_+(\F_t)$ is called 
	an {\em $\mathcal F_t$-conditional convex deviation measure}
	if it is {\em normalized} ($D_t(0) = 0$)
	and satisfies (D1)-(D3).
\end{definition}
By postulating convexity in the sequel instead of (D3')-(D4') our dynamic theory will be richer and include more examples.
In the analysis often also a continuity condition is imposed, the conditional version of which is given as follows:
\begin{itemize}
	\item[{\rm (D4)}] Continuity: 
	{\it If $X^n$ converges to $X$ in $L^2(\F_T)$ then
		$  D_t(X) = \lim_{n}D_t(X^n)$. }
\end{itemize}

\begin{Remark}
A typical example of a deviation measure satisfying (D1)-(D4) would be to 
identify risk with variance and to define
$$D_t(X):=\text{Var}_t(X)=\E{(X-\E{X|\F_t})^2|\F_t}.$$
\end{Remark}

\begin{Remark}
\label{remaconvexriskmeasure}
As mentioned in the introduction the axiomatic development of the theory of 
deviation measures in Rockafellar {\emph et al.} was inspired by the axiomatic 
development of the theory of convex risk measures. Mappings $\p_t:L^2(\F_T)\to 
L^2_+(\F_t)$ are a family of dynamic convex risk measures if the following 
properties are satisfied:
\begin{itemize}
\item[(R1)] {\it Cash Risklessness}: For all $m\in
L^\infty(\F_t)$ we have $\p_t(m)=-m$.
\item[(R2)] $\F_t$-{\it Convexity}: For $X,Y\in L^2(\F_T)$
$\p_{t}(\l X+(1-\l)Y)\leq \l \p_{t}(X)+(1-\l)\p_{t}(Y)$ for all $\l\in 
L^\infty(\F_t)$ such that $0\leq\l\leq1$.
\item[(R3)] {\it Monotonicity:} If $X,Y\in L^2(\F_T)$ and $X\leq Y$ then 
$\p_t(X)\geq \p_t(Y).$
\item[(R4)] \emph{$L^2$-Continuity:} If $X^n$ converges to $X$ in $L^2$, then 
also $\lim_{n}\p_t(X^n) = \p_t(X)$ for all $t$.
\end{itemize}
Axiom (R1) gives the interpretation of $\p$ as a \emph{capital reserve}. Axiom 
(R2) is interpreted similarly as before. The continuity axiom (R4) is often 
also replaced with a lower semi-continuity axiom. Monotonicity (R3) is an axiom 
which does not make sense for deviation measures since for instance $D_t(m)=0$ 
for all constants $m$.
\end{Remark}

Note that (D2) ($\F_t$-convexity) implies that the following property holds:

\be
\label{localproperty}
D_t(I_A
X_1+I_{A^c}X_2)=I_AD_t(X_1)+I_{A^c}D_t(X_2) \mbox{ for all }
X\in L^2(\F_T)\mbox{ and }A\in \F_t.
\ee

Indeed, if $D_t$ is convex, then we have for $A\in\F_t$ that clearly $D_t(1_A 
H_1 + 1_{A^c} H_2) \le 1_A D_t (H_1) + 1_{A^c}D_t(H_2).$ In particular, $
1_A D_t(1_A H_1 + 1_{A^c} H_2) \le 1_A D_t(H_1) .$ The other direction follows 
by setting $\tilde{H} = 1_A H_1 + 1_{A^c} H_2$. Then as before
$$
1_A D_t(H_1)=1_A D_t(1_A
\tilde{H} + 1_{A^c} H_1) \le 1_A D_t(\tilde{H}) .$$ Switching the role of $H_1$ 
and $H_2$ yields then the desired conclusion. Note that (\ref{localproperty}) 
also implies that 
\be
D_t(I_AX) = D_t(I_AX + 0I_{A^c}) = I_AD_t(X) + I_{A^c}D_t(0) = I_AD_t(X).
\label{local2}
\ee

Now in a theory of risk in a dynamic setting one needs to specify how the evaluation of risk tomorrow affects the evaluation of risk today.
Intuitively it seems appealing to relate the overall deviation to 
an expectation to the fluctuations we expect after tomorrow plus the 
fluctuations happening until tomorrow. To be precise we will postulate that

\begin{itemize}
\item[(D5)] \emph{Recursive Property:} For $X\in L^2(\F_T)$
$D_{t}(X)=D_t(\E{X|\F_s})+\E{D_s(X)|\F_t}$ for all $t,s\in [0,T]$ with $t\leq 
s$.
\end{itemize}

Obviously, the recursive property corresponds to the conditional variance 
formula. 
This axiom was used in Pistorius and Stadje~(2017). The conditional variance formula was recently used for instance in Basak and Chabakauri (2010), Wang and Forsyth (2011), Li et al.(2012) and Czichowsky (2013) in the context of time-consistent dynamic programming principles.

\begin{definition}
	A family $(D_t)_{t\in [0,T]}$ is called a {\em dynamic deviation measure} if 
	$D_t$, $t\in [0,T]$, are {$\mathcal F_t$-conditional  deviation measures} 
	satisfying (D4) and (D5).
\end{definition} 

\begin{Remark}
As $D(X)\geq 0$, (D5) clearly implies that $(D_s(X))$ is a
supermartingale. In particular, it always has a c\`adl\`ag
modification.
\end{Remark}
\begin{Remark}
$D_s$ is completely determined by $D_0.$ This is seen 
as follows:
Suppose that besides $D_s(X)$ there exists square integrable
$\F_s$-measurable random variables $D'_s(X)$ satisfying (D5) for all $X$. Fix 
$X$ and denote
the $\F_s$-measurable set $A'$ by $A':=\{D'_s(X)> D_s(X)\}$.
If we by contradiction assume that $A'$ has measure non-zero
then by (\ref{local2}) and (D3)
$$\E{I_{A'} D_s(X)}=\E{D_s(I_{A'} X)}=D_0(I_{A'} X)-D_0(\E{I_{A'} X|\F_s})
=\E{D'_s(I_{A'} X)}=\E{I_{A'} D'_s(X)},$$
which is a contradiction to the definition of the set $A'$. That the
set $\{D'_s(X)< D_s(X)\}$ must have measure zero as well is seen similarly.
\end{Remark}

The following proposition is also shown in Pistorius and Stadje~(2017) in a more resticted axiomatic setting. The proof is included for convenience of the reader.

\begin{proposition}\label{Propositionadditive} 
	Let $I:=\{t_0, t_1, \ldots, t_n\}\subset[0,T]$ be strictly ordered. 
	$D = (D_t)_{t\in I}$ satisfies (D1)--(D3) and (D5)
	if and only if for some collection $\tilde D = (\tilde D_t)_{t\in I}$ of conditional deviation measures 
	we have 
	\begin{equation}\label{DDH}
	D_t(X) = \E{\left. \sum_{t_i\in I: t_i\ge t} 
		\tilde D_{t_i}\left(\E{X|\mathcal F_{t_{i+1}}} - \E{X|\mathcal F_{t_{i}}}\right)\right|\mathcal F_t}, \quad t\in I,\ 
	X\in L^2(\mathcal F_T).
	\end{equation}
	In particular, a dynamic deviation measure $D$ satisfies \eqref{DDH} with $\tilde D_{t_i} = D_{t_i}$, $t_i\in I$.
\end{proposition}
\begin{proof}{}%
	\ `$\Leftarrow$': We will only show that $D_t$ satisfies (D5), as 
	it is clear that (D1)--(D3) are satisfied. Let $X\in L^2(\mathcal F_T)$ and note that 
	as $\tilde D_t$, $t\in I$, satisfy (D1) and (D5) we have for any $s,t\in I$ with $s>t$ 
	that $D_t(\E{X|\mathcal F_s}) = \sum_{t_i\in I:t\leq t_i< s}\E{\tilde D_{t_i}(\E{X|\mathcal F_{t_{i+1}}})|\F_t}$.
	Thus, we have that $D_t(X)$ is equal to
	$$
	\sum_{t_i\in I:t\leq t_i< s}\E{\left.\tilde D_{t_i}(\E{X|\mathcal F_{t_{i+1}}})\right|\F_t} 
	+ \sum_{t_i\in I:s\leq t_i}\E{\left.\tilde D_{t_i}(\E{X|\mathcal F_{t_{i+1}}})\right|\F_t}
	= D_t(\E{X|\mathcal F_s}) + \E{D_s{(X)}|\mathcal F_t}.
	$$
	`$\Rightarrow$': 
	For $X\in L^2(\mathcal F_T)$ and $t_{i-1}\in I$, $i\ge 1$, 
	we have by (D5) and (D1)
	\begin{align}\nonumber
	D_{t_{i-1}}(X) &= D_{t_{i-1}}(\E{X|\F_{t_{i}}}) + \E{D_{t_{i}}(X)|\F_{t_{i-1}}}\\
	&= D_{t_{i-1}}(\E{X|\F_{t_{i}}} - \E{X|\mathcal F_{t_{i-1}}} ) + \E{D_{t_{i}}(X)|\F_{t_{i-1}}}.
	\label{pdp}
	\end{align}
	An induction argument based on \eqref{pdp} then yields 
	that \eqref{DDH} holds with $\tilde D_t=D_t$, $t\in I$.
\end{proof}
\setcounter{equation}{0}

\section{Distribution Invariant Deviation Measures} \label{Section Distribution Invariant Deviation Measures}

The next result investigates the question what happens if we impose 
additionally to axioms (D1)-(D5) \emph{the property of distribution invariance}. A dynamic deviation measure $D$ is distribution invariant if $D_0(X_1) = D_0(X_2)$ whenever $X_1$ and $X_2$ have the same distribution. Distribution invariance is a property which is often not satisfied in a finance context when it comes to evaluation and risk analysis. The reason is that the value of a payoff may not only depend on the nominal discounted value of the payoff itself but also on the whole state of the economy or the performance of the entire financial market. For instance, in no-arbitrage pricing scenarios are additionally weighted with a (risk neutral) density so that the value of a certain payoff in a certain scenario depends not only on the frequency with 
which the corresponding scenario occurs but also on the state of the whole economy. Also in most asset pricing models in finance, not only the distribution of an asset matters but also its correlation to the whole market portfolio. However, for deviation measures distribution invariance is a convenient property as it enables the agent to focus only on the end-distribution of the payoff (which often is known explicitly or can be simulated through Monte-Carlo methods). There are many distribution invariant static deviation measures but it is a priori not clear if apart from variance there are other dynamic deviation measures belonging to this class. The next theorem shows that this is actually not the case and that variance is the only dynamic distribution invariant deviation measure. 
This result can also serve as justification for using variance as a dynamic deviation measure. Namely, a decision maker who believes in axioms (D1)-(D5) and distribution invariance necessarily has to use variance.
For this results we will assume that the probability space is rich enough to 
support a one-dimensional Brownian motion.
\begin{theorem}
	\label{distributioninvariant}
	A dynamic dynamic deviation measure $D$ is distribution invariant if and only 
	if $D$ is a positive multiple of the conditional variance, i.e., there exists an 
	$\alpha > 0$ such that
	
	$$D_t(X)=\alpha {\rm Var}_t (X) .$$
\end{theorem}

For the proof of Theorem \ref{distributioninvariant} we will need the following 
lemma:
\begin{lemma} \label{indY}
	Suppose that $D_t$ is a family of dynamic distribution invariant deviation 
	measures and that $Y$ is independent of $\F_t$. Then, $D_t(Y)$ is constant and 
	$$D_t(Y)=D_0(Y).$$
\end{lemma}
\begin{proof}
	The case that $t=0$ is trivial. So let us assume that $t>0$. Suppose then that 
	$D_t(Y)\neq \text{constant}$. Choose sets $A, A'\in \F_t$ with $P(A)=P(A')>0$ 
	such that $D_t(Y)(\omega) >\ D_t(Y)(\omega')$ for all $\omega\in A, \omega'\in 
	A'$. That this is possible can be seen as follows. Since $D_t(Y)$ is not constant clearly there exists $c \in \mathbb{R}$ such that the set $B:= \{D_t(Y) \geq c\}$ has probability strictly between zero and one. Assume without loss of generality that $\mathbb{P}(B) \geq \frac{1}{2}$. Set $A' = B^c$ (the complement of $B$) and note that $$\mathbb{P}(A') \leq \frac{1}{2} \leq \mathbb{P}(B).$$ Define $\tilde{A}^r := B \cap \{U \leq r\}$ with $r \in [0,1]$ and $U = F_{B_t}(B_t) \sim Unif[0,1]$ where $F_{B_t}$ is the cdf of the Brownian motion $B_t$. By definition $U$ is $\F_{t}$-measurable. Clearly $r \to \mathbb{P}(\tilde{A}^r)$ is a continuous function taking all values between $[0,\mathbb{P}(B)]$. In particular, there exists $r_0$ such that $\mathbb{P}(\tilde{A}^r) = \mathbb{P}(A')$. Setting $A = \tilde{A}^{r_0}$ completes the argument. \\
	
	Next note that by independence $I_A Y \stackrel{D}{\sim} I_{A'}Y$. However,
	
	\begin{align*}
		D_0(I_A Y) &= D_0 ( \E{I_AY|\F_t}) + \E{D_t(I_A Y)} \\
		&= D_0 (I_A \E{Y|\F_t}) + \E{I_AD_t( Y)} \\
		&= D_0 (I_A \E{Y}) + \E{I_A D_t(Y)} \\
		&< D_0 (I_A \E{Y}) + \E{I_{A'} D_t(Y)}\\
		&=D_0 (I_{A'} \E{Y}) + \E{I_{A'} D_t(Y)} = D_0 (\E{I_{A'}Y|\F_t}) + 
		\E{D_t(I_{A'} Y)} = D_0(I_{A'} Y),
	\end{align*}
	which is a contradiction to the distribution invariance of $D_0$. So indeed 
	$D_t(Y)$ is constant. Finally, by the first part of the proof
	
	\begin{align*}
		D_0(Y) &= \E{D_t(Y)} + D_0(\E{Y|\F_t}) \\
		&= \E{D_t(Y)} + D_0(\E{Y}) = \E{D_t(Y)} = D_t(Y).
	\end{align*}
\end{proof}

\textit{Proof of Theorem \ref{distributioninvariant}.}
Since $D_t$ is uniquely determined by $D_0$ it is sufficient to prove the 
theorem for $t=0$. Let us first show that the theorem holds for $X$ having a 
normal distribution. Let $Z$ be a standard normally distributed random 
variable. Define $f(\sigma)=D_0(\sigma Z)$ with $\sigma\in\mathbb{R}.$ By 
assumption there exists an adapted Brownian motion, say $(B_t)_{0\leq t\leq 
	T}$. It is then for $0\leq t \leq T$

\begin{align*}
	f(\sigma\sqrt{t}) &= D_0(\sqrt{t}\sigma Z) \\
	&= D_0(\sigma B_t) \\
	&= \sum_{i=0}^{n-1} \E{D_{ti/n}(\sigma\Delta B_{t(i+1)/n})} = \sum_{i=0}^{n-1} 
	D_0(\sigma\Delta B_{t(i+1)/n}) = n D_0\left(\frac{\sqrt{t}\sigma 
		Z}{\sqrt{n}}\right) = n f\left(\frac{\sqrt{t}\sigma}{\sqrt{n}}\right),
\end{align*}
where we set $ \Delta B_{t(i+1)/n}:= B_{t(i+1)/n}- B_{ti/n}.$
It follows that 
$f\left(\frac{\sqrt{t}\sigma}{\sqrt{n}}\right)=\frac{f(\sqrt{t}\sigma)}{n}$. 
Arguing similarly as before
we also get for $k\in\mathbb{N}$ with $\frac{k}{n}\leq \frac{T}{t}$

\begin{align*}
	f\left(\sqrt{\frac{k}{n}t}\sigma\right) &=D_0(\sigma B_{kt/n})\\
	&=\sum_{i=0}^{k-1} \E{D_{ti/n}(\sigma\Delta B_{t(i+1)/n})} \\&
	= k D_0(\sigma B_{t/n}) = k D_0\left(\frac{\sigma B_t}{\sqrt{n}}\right) = \frac{k}{n} D_0(\sigma B_t) = \frac{k}{n}f(\sigma \sqrt{t}).
\end{align*}
By (D4) we have that $f$ is continuous. Therefore, for all $0\leq \lambda \leq 
\frac{T}{t}$, $f(\lambda \sigma \sqrt{t}) = \lambda^2 f(\sigma \sqrt{t})$ for 
any $\sigma\in\mathbb{R}$. Setting for arbitrary $x\in\mathbb{R}$, $\sigma = 
x/\sqrt{t}$, we get that $f(\lambda x) = \lambda^2 f(x)$ for all $0\leq \lambda 
\leq \frac{T}{t}$
with $t\in[0,T]$. Choosing $t$ arbitrary small, we may conclude that $f(\lambda 
x) = \lambda^2 f(x)$ for all $\lambda \in \mathbb{R}_+.$ Hence, if we define 
$\alpha:=f(1) > 0$ we have that

$$D_0(\sigma Z)=D_0(|\sigma| Z)=f(|\sigma |)=\sigma^2 \alpha=\alpha {\rm Var}(Z) , $$
where the first equality follows by the distribution invariance of $D_0.$
Next let us show that for simple functions of the form $X = \left( (h_i 
I_{(t_i, t_{i+1}]})\cdot B\right)_{t_i,t_{i+1}}$
with $h_i=\sum_{j=1}^{m} c_j I_{A_j}$, $c_j\in\mathbb{R}^d$, and disjoint sets 
$A_j\in \F_{t_i}$ for $j=1,\ldots,m$ we have $$D_0(X) = \alpha {\rm Var}(X).$$ 
It is
\begin{align*}
	D_0(X)&=D_0(h_{t_i}\Delta B_{t_{i+1}} ) \\
	&= \E{\sum_{j=1}^mI_{A_j}D_{t_i}(c_j \Delta B_{t_{i+1}} )} \\
	&=\alpha\, \E{\sum_{j=1}^m I_{A_j} c_j^2 (t_{i+1}-t_i) } = 
	\alpha\,\E{h_i^2(t_{i+1}-t_i)} = \alpha {\rm Var}(X),
\end{align*}
where we used Lemma \ref{indY} in the third equation to argue that $D_{t_i}(c_j 
\Delta B_{t_{i+1}} )=D_{0}(c_j \Delta B_{t_{i+1}} )=c_j^2 (t_{i+1}-t_i) $.
For $X = \left( (h_i I_{(t_i, t_{i+1}]})\cdot B\right)_{t_i,t_{i+1}}$ with 
general $h_i\in L^2_d(\F_{t_i},\mathbb{P})$ choose simple function $h_i^n$ 
converging to $h_i$ in $L^2$ and define $X^n = \left( (h_i^n I_{(t_i, 
	t_{i+1}]})\cdot B\right)_{t_i,t_{i+1}}$. Using $L^2$-continuity of $D_0$ we may 
conclude that $$ D_0(X) = \lim_n D_0(X^n) = \lim_n \alpha {\rm Var}(X^n) = 
\alpha {\rm Var}(X).$$

Next note that for simple functions of the form
$X = \sum_{i=1}^l\left(( h_i I_{(t_i, t_{i+1}]})\cdot B\right )_{t_i,t_{i+1}}$ 
for $l\in \mathbb{N}$, $h_i\in L^2_d(\F_{t_i},d\mathbb{P})$ we have

\begin{align*}
	D_0(X) &= \sum_{i=1}^l \E{D_{t_i}\Big(\left( (h_i I_{(t_i, t_{i+1}]})\cdot 
		B\right)_{t_i,t_{i+1}}\Big )}\\
	&= \sum_{i=1}^l D_{0}\bigg(\left( (h_i I_{(t_i, t_{i+1}]})\cdot 
	B\right)_{t_i,t_{i+1}} \bigg)= \sum_{i=1}^l \alpha {\rm Var}\left( ((h_i 
	I_{(t_i, t_{i+1}]})\cdot B)_{t_i,t_{i+1}} \right)= \alpha {\rm Var}(X),
\end{align*}
where we used Proposition \ref{Propositionadditive} in the first and second equation. 
Therefore, $D_0(X)=\alpha {\rm Var}(X)$ for all simple functions $X$. Using the 
$L^2$-continuity of $D_0$ and $\alpha {\rm Var}(X)$ as before, we get that 
equality actually holds for all $X\in L^2(\F_T^{B})$ with $\F_T^{B}$ being the 
completion of the $\sigma$-algebra generated by $(B_t)_{0\leq t \leq T}$.
Next, take a general $X\in L^2(\F_T)$. Define the uniform $[0,1]$ distributed 
random variable $U = F_{B_T}(B_T)$ where $F_{B_T}$ is the cdf of $B_T$. Set $X' 
= q_X(U) \stackrel{D}{=} X$. Then clearly $X'$ is $\F_T^{B}$-measurable. 
Therefore,

\begin{align*}
	D_0(X) &= D_0(X') = \alpha {\rm Var}(X') = \alpha {\rm Var}(X).
\end{align*}
This proves the theorem.
\qed

\begin{remark}
	Our proof also works for distribution invariant dynamic deviation measures $(D_t(X))_{t \in \mathbb{N}_0}$, i.e., for distribution invariant dynamic deviation measures only defined (and satisfying (D1)-(D5)) on $t \in \mathbb{N}_0$.
\end{remark}
Kupper and Schachermayer~(2009) showed that a dynamic convex risk measure is 
law-invariant if and only if there exists there exists $\gamma\in[0,\infty]$ 
such that

\be
\label{entropy}
\p_t(X)=\frac{1}{\gamma}\E{\exp(-\gamma X)|\F_t}.
\ee
The limiting case $\gamma=0$ and $\gamma=\infty$ are identified with the 
conditional expectation and the essential supremum respectively. Related 
results are also known for insurance premiums, see Gerber~(1974) and the 
references given in the introduction. 
\setcounter{equation}{0}

\section{Risk-sharing in continuous time} \label{Section Examples and Applications}

\label{sec4}
In this chapter, we assume that the probability space $(\Omega,\F,\mathbb P)$ is equipped with
(i) a standard $d$-dimensional Brownian motion
$W=(W^1,\ldots,W^d)^\intercal$ and (ii) a Poisson
random measure $N(\td t\times \td x)$ on
$[0,T]\times\mathbb{R}^k\setminus\{0\},$ independent of $W$,
with intensity measure $\hat{N}(\td t\times \td x)=\nu(\td x)\td t$, 
where the L\'{e}vy  measure $\nu(\td x)$ satisfies the integrability condition
$$ \int_{\mathbb{R}^k\setminus\{0\}} (|x|^2\wedge 1) \nu(\td x)<\infty,$$
and let $\tilde{N}(\td t \times \td x):=N(\td t \times \td x)-\hat{N}(\td t \times \td x)$ denote the compensated Poisson random measure.
Further, let $\mathcal{U}$ denote the Borel sigma-algebra induced by the $L^2(\nu(\td x))$-norm, $(\F_{t})_{t\in[0,T]}$ the right-continuous completion of the filtration
generated by $W$ and $N$, and $\mathcal{P}$ and $\mathcal O$ the
predictable and optional sigma-algebras on $[0,T]\times\Omega$ with respect to
$(\F_{t})$.  We denote by $ L_d^2(\mathcal{P},\td\mathbb P\times \td t)$ 
the space of all predictable $d$-dimensional processes that are square-integrable with respect to the measure 
$\td\mathbb P\times
\td t$ and we let $\mathcal S^2 = \left\{Y\in\mathcal O: \E{\sup_{0\leq t\leq T}|Y_t|^2} < \infty\right\}$
denote the collection of square-integrable c\`{a}dl\`{a}g optional processes. 
Further, let $\mathcal{B}(\R^k\setminus\{0\})$ be the Borel sigma-algebra on $\R^k\setminus\{0\}.$
For any $X\in L^2(\F_T)$ we denote by $(H^X,\tilde H^X)$ the unique pair of predictable processes with 
$H^X\in L_d^2(\mathcal{P},\td\mathbb P\times \td t)$ and $\tilde{H}^X
\in L^2(\mathcal{P}\times \mathcal{B}(\R^k\setminus\{0\} ),\td\mathbb P\times \td t \times \nu(\td x))$, subsequently 
referred to as the \emph{representing pair} of $X$, satisfying\footnote{See {\em e.g.}\, Theorem III.4.34 in Jacod and Shiryaev (2013)}

\begin{equation}\label{mrep}
X=\E{X}+\int_0^T H^X_s
\td W_s+\int_0^T\int_{\R^k\setminus\{0\}}\tilde{H}_{s}^X(x)\tilde{N}(\td s \times \td x) ,
\end{equation}
where $\int_0^T H^X_s \td W_s:=\sum_{i=1}^d \int_0^T H^{X,i}_s \td W^i_s$. We call a $\mathcal{P}\otimes \mathcal{B}(\mathbb{R}^d)\otimes \mathcal{U}$-measurable function $$
\begin{array}{rlclclclll}
g:&[0,T] &\times &\Omega &\times &\mathbb{R}^{d} &\times & L^{2}(\nu(\td x)) &\rightarrow &\R_+\\
&(t, & &\omega , & &h, & & \tilde{h}) &\longmapsto &g(t,\omega ,h,\tilde{h})%
\end{array}
$$
a {\em driver function} if for $\td\mathbb P\times \td t$ a.e. $(\omega,t)\in\Omega\times[0,T]$: $g$ is zero if and only if $(h,\tilde{h})=0$, and $g$ is convex and lower semi-continuous in $(h,\tilde{h})$. It is shown in Theorem 4.1 in Pistorius and Stadje~(2017), that a dynamic deviation measure satisfying (D1)-(D5) is equivalent to the existence of a driver function $g$, such that $$D_t(X) = Y_t, \quad\quad X\in L^2(\mathcal F_T),$$ 
where $(Y,Z,\tilde{Z})$ 
is the unique square integral solution of the SDE given by 
\begin{eqnarray}
\label{bsde}
\td Y_t &=& - g(t,H^X_t,\tilde{H}^X_{t})\td t+Z_t \td W_t+\int_{\mathbb R^k\backslash\{0\}} 
\tilde{Z}_t(x)\tilde{N}(\td t \times \td x),\quad t\in[0,T),\\
Y_T &=& 0.
\label{bsde2}
\end{eqnarray} 
Equivalently, we can write \begin{equation}
\label{intrep}
D_t(X) = \E{\int_t^Tg(s,H_s^X,\tilde{H_s^X}) \td s \bigg| \F_t}.
\end{equation}
Hence, any deviation measure admits an integral representation (\ref{intrep}) in terms of a function $g$. \\

Now let $\p$ be the entropic risk measure defined by (\label{entropy}) which is the only time-consistent law-invariant dynamic convex risk measure. By well known results from the backward stochastic differential equation (BSDE) literature (see for instance Barrieu and El Karoui~(2009) or Pelsser and Stadje~(2014)) (\ref{entropy}) entails in a Brownian-Poisson filtration, that there exists predictable square integrable $(Z,\tilde{Z})$ such that 
$(\p_t(X))_{0\leq t\leq T}$ satisfies $\rho_T(X) = X$ and 
$$d\p_t(X)=-\Big(\frac{1}{2\gamma}|Z_s|^2+\IR 
h(\tilde{Z}_s(x))\nu(dx)\Big)ds+Z_sdW_s+\tilde{Z}_s(x)\tilde{N}(ds,dx) ,$$
where $h(x)=\frac{1}{\gamma}(\exp\{\gamma x \}-\gamma x-1)$.
Now since the conditional variance process corresponds to a quadratic driver, 
Theorem \ref{distributioninvariant} entails that a distribution invariant deviation measure satisfies

$$dD_t(X)=-\alpha\Big(|H^X_s|^2+\IR |\tilde{H}^X_s(x)|^2 \nu(dx)\Big)ds+Z_sdW_s+\tilde{Z}_s(x)\tilde{N}(ds,dx) .$$
It is interesting to note that this means that in a case of purely Brownian 
filtration without jumps for both dynamic risk measure and for dynamic deviation 
measures distribution invariance both lead to a quadratic driver (penalty) function with the difference that for a dynamic risk measure the Malliavin derivative of the evaluation itself is squared, while for a dynamic deviation measure the Malliavin derivative of the terminal payoff is squared. However, for the Malliavin derivatives of the jump parts distribution invariance entails different kind of penalizing (namely exponential in the one and quadratic in the other case). The reason is that a Taylor approximation cannot be applied to the infinitesimal jump parts because of the discontinuities.





\begin{Example}\label{Dl} The family of $g$-deviation measures	with driver functions given by
	\begin{equation}\label{gl}
	g_{c,d}(t,h,\tilde h) = c\, |h| + d\, \sqrt{\IR |\tilde{h}(x)|^2
		\nu(\td x)}, \quad c,d\in\mathbb R_+\backslash\{0\}, 
	\end{equation}
	corresponds to a measurement of the risk of a random variable $X\in L^2(\F_T)$ by the integrated multiples 
	of the local volatilities of the continuous and discontinuous martingale parts in 
	its martingale representation~\eqref{mrep}. 
\end{Example}
\begin{Example}\label{D2}
	In the case of a $g$-deviation measure with driver function given by 
	$$
	g(\omega, t,h,\tilde{h}) = {\it CVaR}^{\nu}_{t,a}(\tilde h), \qquad a\in(0,\nu(\mathbb R^k\backslash\{0\})),
	$$
	the risk is measured in terms of the values of the (large) jump sizes under $CVaR^\nu_{t,a}$.
	Here $CVaR^{\nu}_{t,a}(\tilde h) = \frac{1}{a}\int_0^a 
	VaR^{\nu}_{t,b}(\tilde h)\td b
	$
	is given in terms of the left-quantiles $VaR^{\nu}_{t,a}(\tilde h)$, $a\in(0,\nu(\mathbb R^k\backslash\{0\}))$ 
	of $h(J)$ under the measure $\nu(\td x)$, that is,
	$$
	VaR^{\nu}_{t,a}(\tilde h) := VaR^{\nu}_{a}(h(J)) := \sup\{y\in\mathbb R: 
	\nu(\{x\in \mathbb R^k\backslash\{0\}: \tilde h(x) < -y\}) < a \}.
	$$
\end{Example}

Next, we analyze the problem of optimal risk-sharing. Suppose that we have two agents holding square integrable positions $X_A$ and 
$X_B$ and using dynamic deviation measures $D^A$ and $D^B$, respectively. Agent $A$ evaluates her risk, say $X$, by $$U^A_t(X) = \E{X|\F_t} - D_t^A(X).$$ Agent $B$ evaluates her risks similarly. Suppose the agents are allowed to set up contracts with each other specifying in every scenario $\omega \in \Omega$ a payment $Y'(\omega)$. We will refer to $Y'$ also as a payoff and assume that agent $A$ exchanges $Y'$ for a price $\pi_{Y'}$. Let us assume further that only square-integrable payoffs can be traded (or in other words exchanged). Exchanging $Y'$ for a price of $\pi_{Y'}$ agent $A$ can then reduce her risk optimally by seeking
\begin{align}\label{mini2}
	\arg\sup_{Y'\in L^2}U^A_0(X_A-(Y'-\pi_{Y'})).
\end{align}

Note that this is a generalization of a mean-variance optimization problem. For Agent $B$ to enter the transaction her utility should at least be as high as before. Therefore, Agent $A$ is under the constraint

\begin{equation}\label{constraint}
	U^B((Y'-\pi_{Y'})+X_B) \geq U^B(X_B)
\end{equation}

For the study of similar problems in the case of other risk measures, see the references given in the introduction. The following theorem gives a complete solution to the problem.
\begin{theorem}
	\label{risksharing}
	Let $g_A$ and $g_B$ be the driver function corresponding to $D^A$ and $D^B$ 
	respectively. Then the problem in (\ref{mini2}) has a solution $\tilde{Y}^*$ if and only if 
		there exist
		$$(H^*_t,\tilde{H}^{*}_t)\in \argmin_{H\in L^2_d(\td P),\tilde{H}\in L^2(\td P\times 
			\nu(\td x))}\{g_A(t,H^{X_A+X_B}_t
		-H,\tilde{H}^{X_A+X_B}_t-\tilde{H})+g_B(t,H,\tilde{H})\},$$
		such that $(H^*_s)_{0\leq s\leq T}$ is in $L^2_d(\td P\times \td s)$ and 
		$(\tilde{H}^*_s)_{0\leq s\leq T}$ is in $L^2(\td P\times \td s \times \nu(\td x))$. In 
		this case an optimal risk transfer is given by $\tilde{Y}^*= \int_{0}^{T} H^*_s \td W_s 
		+ \int_{0}^{T}\int_{\mathbb{R}^k\setminus\{0\}}\tilde{H}^*_s(x) \tilde{N}(\td s, \td x) - X_B$.

\end{theorem}

\textit{Proof of Theorem \ref{risksharing}.}
By translation invariance ((D1)) we obtain from (\ref{constraint}) $\pi_{Y'} = \E{Y'} - D^B_t(X_B+Y') + D^B_t(X_B)$. Again using translation invariance in (\ref{mini2}) the optimal risk allocation is given by

\begin{align}
	\arg\essinf_{Y'\in L^2}\{D^A_0(X_A-Y')+D^B_0(X_B+Y')\}&=\arg\essinf_{Y\in L^2}\{D^A_0(X_A+X_B-Y)+D^B_0(Y)\}\nonumber.
\end{align}
The second equation may be seen by redefining $Y:=Y'+ X_B$. We define further

\begin{align}\label{min}
	D_t(X_A+X_B) :=\essinf_{Y\in L^2}\{D^A_t(X_A+X_B-Y)+D^B_t(Y)\}. 
\end{align}

Theorem \ref{risksharing} follows then from the Proposition below. \qed

\begin{proposition}
	\label{risksharing2}
	Let $g_A$ and $g_B$ be the driver function corresponding to $D^A$ and $D^B$ 
	respectively. Then
	\begin{itemize}
		\item[(i)] $D_t$ defined in (\ref{min}) is a dynamic deviation measure with 
		driver function given by $$g(t,h,\tilde{h}):=(g_A\square 
		g_B)(t,h,\tilde{h}):=\essinf_{z\in\mathbb{R}^d,\tilde{z}\in 
			L^2(\nu)}\{g_A(t,h-z,\tilde{h}-\tilde{z})+g_B(t,z,\tilde{z})\}.$$
		\item[(ii)] The infinum in (\ref{min}) is attained in $Y^*$ if and only if 
		there exist
		$$(H^*_t,\tilde{H}^{*}_t)\in \argmin_{H\in L^2_d(\td P),\tilde{H}\in L^2(\td P\times 
			\nu(\td x))}\{g_A(t,H^{X_A+X_B}_t
		-H,\tilde{H}^{X_A+X_B}_t-\tilde{H})+g_B(t,H,\tilde{H})\},$$
		such that $(H^*_s)_{0\leq s\leq T}$ is in $L^2_d(\td P\times \td s)$ and 
		$(\tilde{H}^*_s)_{0\leq s\leq T}$ is in $L^2(\td P\times \td s \times \nu(\td x))$. In this case an optimal $Y$ in (\ref{min}) is given by $Y^*= \int_{0}^{T} H^*_s \td W_s 
		+ \int_{0}^{T}\int_{\mathbb{R}^k\setminus\{0\}}\tilde{H}^*_s(x) \tilde{N}(\td s, \td x)$.
	\end{itemize}
\end{proposition}

\textit{Proof of Proposition \ref{risksharing2}.}
Let us start showing (i). Note that

\begin{align}
	&D_t(X_A+X_B)\nonumber\\
	&=\essinf_{Y\in L^2}\{D^A_t(X_A+X_B-Y)+D^B_t(Y)\}\nonumber\\
	&=\essinf_{H^Y\in L^2_d(dP\times ds),\tilde{H}^Y\in L^2(dP\times ds\times 
		\nu(dx))}\nonumber\\
	&\hspace{1cm}\E{\int_t^T 
		[g_A(s,H^{X_A+X_B}_s-H^Y_s,\tilde{H}^{X_A+X_B}_s-\tilde{H}^Y_s)+g_B(s,H^Y_s,\tilde{H}^Y_s)]ds
		|\F_t}\nonumber\\
	&\geq \E{\int_t^T (g_A\square g_B)(s,H^{X_A+X_B}_s,\tilde{H}^{X_A+X_B}_s)ds
		|\F_t}\label{inequ}.
\end{align}
Hence, what is left to show is `$\leq$' in (\ref{inequ}).
By a measurable selection theorem we may choose predictable processes such that

\begin{align*}
	g_A(s,H^{X_A+X_B}_s-H^\varepsilon_s,\tilde{H}^{X_A+X_B}_s-\tilde{H}^{\varepsilon}_s)+g_B(s,H^\varepsilon_s,\tilde{H}^\varepsilon_s)
	\leq
	(g_A\square g_B)(s,H^{X_A+X_B}_s,\tilde{H}^{X_A+X_B}_s)+\varepsilon .
\end{align*}
Thus,
\begin{align*}
	&\E{\int_t^T 
		[g_A(s,H^{X_A+X_B}_s-H^\varepsilon_s,\tilde{H}^{X_A+X_B}_s-\tilde{H}^{\varepsilon}_s)+g_B(s,H^\varepsilon_s,\tilde{H}^\varepsilon_s)]ds|\F_t}
	\\
	&\hspace{1cm}\leq
	\E{\int_t^T (g_A\square 
		g_B)(s,H^{X_A+X_B}_s,\tilde{H}^{X_A+X_B}_s)ds|\F_t}+\varepsilon.
\end{align*}
Choosing $\varepsilon$ arbitrary small we get `$\leq$' in (\ref{inequ}) 
completing the proof of (i).
To show (ii) note that if there exists $H^*$ in $L^2_d(dP\times ds)$ and 
$\tilde{H}^*$ in $L^2(dP\times ds \times \nu(dx))$ such that for Lebegue a.s. 
all $t$

\begin{align}
	\label{conv}
	& g_A(s,H^{X_A+X_B}_s
	-H^*_s,\tilde{H}^{X_A+X_B}_s-\tilde{H}^*_s)+g_B(t,H^*_s,\tilde{H}^*_s)\nonumber\\
	&\hspace{0.cm}=\essinf_{H\in L^2_d(dP),\tilde{H}\in L^2(dP\times \nu(dx))} 
	\{g_A(s,H^{X_A+X_B}_s
	-H,\tilde{H}^{X_A+X_B}_s-\tilde{H})+g_B(t,H,\tilde{H})\},
\end{align}
then by the first part of the proof \begin{align}
	&\E{\int_t^T 
		[g_A(s,H^{X_A+X_B}_s-H^{Y^*}_s,\tilde{H}^{X_A+X_B}_s-\tilde{H}^{Y^*}_s)+g_B(s,H^{Y^*}_s,\tilde{H}^{Y^*}_s)]ds
		|\F_t}\nonumber\\
	&=\E{\int_t^T (g_A\square g_B)(s,H^{X_A+X_B}_s,\tilde{H}^{X_A+X_B}_s)ds
		|\F_t}=D_t(X_A+X_B)\label{DAB},
\end{align}
so that the corresponding $Y^*$ solves the minimization problem. On the other 
hand, assume that no square integrable $(H^*_s,\tilde{H}^*_s)_s$ attain the 
infimum in (\ref{conv}) $dP\times dt$ a.s. Then for any $Y\in L^2$ the 
corresponding $H^Y$ and $\tilde{H}^Y$ from the martingale representation theorem
satisfy \beas
g_A(s,H^{X_A+X_B}_s
-H^Y_s,\tilde{H}^{X_A+X_B}_s-\tilde{H}^Y_s)+g_B(t,H^Y_s,\tilde{H}^Y_s)\geq 
(g_A\square g_B)(s,H^{X_A+X_B}_s,\tilde{H}^{X_A+X_B}_s),
\eeas
with a strict inequality on a nonzero predictable set. But this entails that 
the first equation in (\ref{DAB})(with $t=0$) becomes a strict inequality so that $Y$ can not be a solution to the risk sharing problem.
\qed

The next corollary shows that if $X_A  \neq -X_B + const$ it is never optimal to shift all the risks to one single party. This situation is contrary for instance to decision theories like the dual theory of Yaari (1987). However it is in line with risk-sharing under expected utility, see for instance F{\"o}llmer and Schied (2004) and Boonen (2017).

\begin{corollary} \label{coro1}
	Suppose that $X_A \neq -X_B + const$ and that one of the agents has a driver function $g$ which on a $\td P\times \td t$ non-zero set is differentiable in $(0,0).$  Then this agent after the optimal risk transfer will always keep some residual risk. In other words $\tilde{Y}^* + X_B \neq const$ for agent A and $X_A - \tilde{Y}^* \neq const$ for agent B.
\end{corollary}

\begin{proof}
By Theorem \ref{risksharing} we have for the optimal representing pair $$(H^*_t,\tilde{H}^{*}_t)\in \argmin_{H\in L^2_d(dP),\tilde{H}\in L^2(dP\times \nu(dx))}\{g_A(t,H^{X_A+X_B}-H,\tilde{H}^{X_A+X_B}-\tilde{H})+g_B(t,H,\tilde{H})\}.
$$ Hence, $(H^*_t,\tilde{H}^{*}_t)$ by Corollary 2.4.7 in Zalinescu (2002) must satisfy 
\begin{equation} \label{gradient}
\partial g_A(t,H^{X_A+X_B}_t-H^*_t,\tilde{H}_t^{X_A+X_B}-\tilde{H}_t^*)\cap \partial g_B(t,H^*_t,\tilde{H}^*_t) \neq \emptyset,
\end{equation}
where $\partial g_A$ and $\partial g_B$ are the subgradients\footnote{$z$ is in the subgradient of $f(x)$ where $f$ is a convex function if $f(y) \geq f(x)+z(y-x)$ for all $y$.} (i.e., the generalized derivatives) of $g_A(t,z,\tilde{z})$ and $g_B(t,z,\tilde{z})$ with respect to $(z,\tilde{z})$. Noticing that every driver function is zero in zero and otherwise strictly positive, both drivers have their unique minimum in zero. Hence, the subgradients of $g_A$ and $g_B$ do not contain zero at non-zero points. Furthermore, both subgradients contain zero at $(0,0)$. 
From this it follows that $(H^*_t,\tilde{H}^{*}_t)$ can not be identical zero if $\partial g_B(t,0,0) =\{0\}$ on a non-zero set or equal to $(H^{X_A + X_B}, \tilde{H}^{X_A+X_B})$ if $\partial g_A(t,0,0) =\{0\}$ on a non-zero set. The reason is that by assumption $X_A + X_B \neq const$ and therefore $(H^{X_A + X_B}, \tilde{H}^{X_A+X_B}) \neq (0,0)$. 	
\end{proof}

The next corollary describes the situation where both driver functions are the same up to one parameter which is often interpreted as reflecting the degree of risk tolerance. This result gives a closed formula for the optimal risk exchange, $\tilde{Y}^*$. Similar structures of risk-sharing are shown in Borch (1962) (for expected utility instead of dynamic deviation measures). See also Barrieu and El Karoui (2004,2005,2009) for convex risk measures.

\begin{corollary} \label{coro2}
	Suppose that for certain $\gamma_A,\gamma_B >0$ we have $g_A(t,h,\tilde{h})=\gamma_A g(t,h/\gamma_A,\tilde{h}/\gamma_A)$ 
	and $g_B(t,h,\tilde{h})=\gamma_B g(t,h/\gamma_B,\tilde{h}/\gamma_B)$ for all $(h,\tilde{h})$. Then 
	$\tilde{Y}^*=\frac{\gamma_B}{\gamma_A+\gamma_B}X_A-\frac{\gamma_A}{\gamma_A+\gamma_B}X_B$.
\end{corollary}
\setcounter{equation}{0}

\begin{proof}
By the last proof $(H^*, \tilde{H}^*)$ is uniquely characterized by (\ref{gradient}) which becomes
$$
\partial g\bigg(t,\frac{H^{X_A+X_B}_t-H^*_t}{\gamma_A},\frac{\tilde{H}_t^{X_A+X_B}-\tilde{H}_t^*}{\gamma_A}\bigg)\cap \partial g\Big(t,\frac{H^*_t}{\gamma_B},\frac{\tilde{H}^*_t}{\gamma_B}\Big) \neq \emptyset.
$$
Clearly, this holds if we choose $(H^{*}, \tilde{H}^{*})$ to be
$$(H^{\frac{\gamma_B}{\gamma_A+\gamma_B}(X_A + X_B)}, \tilde{H}^{\frac{\gamma_B}{\gamma_A+\gamma_B}(X_A + X_B)})=\bigg(\frac{\gamma_B}{\gamma_A+\gamma_B} H^{X_A + X_B}, \frac{\gamma_B}{\gamma_A+\gamma_B}\tilde{H}^{X_A+X_B}\bigg).$$ This proves the corollary.
\end{proof}

{\footnotesize 

\end{document}